\documentclass[envcountsame]{llncs}
\usepackage{mdwlist}
\usepackage{amssymb,amsbsy,latexsym}
\usepackage{amsmath}
\usepackage{graphics, subfigure, float}
\usepackage[applemac]{inputenc}
\usepackage[american]{babel}
\usepackage[section]{algorithm}
\usepackage[noend]{algpseudocode}
\usepackage{xspace}

\usepackage{amscd}

\usepackage{verbatim}%\usepackage{comment}
\usepackage[dvips]{graphicx,epsfig,color}
\usepackage{pstricks}
\newpsobject{showgrid}{psgrid}{subgriddiv=1,griddots=10,gridlabels=6pt}

\usepackage{tikz}
\usetikzlibrary{snakes}

\providecommand{\R}{\mathbb{R}}

\floatstyle{ruled}
\newfloat{Figure}{t}{fig}
\newfloat{algorithm}{tbp}{loa}
\floatname{algorithm}{Algorithm}
\newfloat{algorithm*}{tbp}{loa}
\floatname{algorithm*}{Algorithm}

\usepackage[displaymath,textmath,sections,graphics, subfigure, floats]{preview}
\PreviewEnvironment{enumerate} % Attention: Do not use comma-seperated lists here
\PreviewEnvironment{itemize}
\PreviewEnvironment{theorem}
\PreviewEnvironment{lemma}
\PreviewEnvironment{myproblem}
\PreviewEnvironment{corollary}
\PreviewEnvironment{abstract}
\PreviewEnvironment{defn}
\PreviewEnvironment{center}
\PreviewEnvironment{pspicture}
\newcommand{\ignore}[1]{}

\newcommand{\fami}{\ensuremath{\mathcal{F}}}

\newcommand{\mst}{{\sc MST}\xspace}
\newcommand{\umst}{{\sc Union MST}\xspace}
\newcommand{\kmst}{{\sc $k$-MST}\xspace}
\newcommand{\ukmst}{{\sc Union $k$-MST}\xspace}
\newcommand{\ikmst}{{\sc Intersection $k$-MST}\xspace}
\newcommand{\ukst}{{\sc Union $k$-Steiner Tree}\xspace}
\newcommand{\st}{{\sc Steiner Tree}\xspace}
\newcommand{\kst}{{\sc $k$-Steiner Tree}\xspace}

\newcommand{\setc}{{\sc Set Cover}\xspace}
\newcommand{\ksc}{{\sc $k$-Set Cover}\xspace}
\newcommand{\uksc}{{\sc Union $k$-Set Cover}\xspace}
\newcommand{\iksc}{{\sc Intersection $k$-Set Cover}\xspace}
\newcommand{\tsp}{{\sc TSP}\xspace}
\newcommand{\ktsp}{{\sc $k$-TSP}\xspace}
\newcommand{\uktsp}{{\sc Union $k$-TSP}\xspace}

\newcommand{\fl}{{\sc Facility Location}\xspace}
\newcommand{\ufl}{{\sc Nonmetric Facility Location}\xspace}
\newcommand{\mfl}{{\sc Metric Facility Location}\xspace}

\newcommand{\kufl}{{\sc $k$-Nonmetric Facility Location}\xspace}
\newcommand{\kmfl}{{\sc $k$-Metric Facility Location}\xspace}
\newcommand{\ukufl}{{\sc Union $k$-Nonmetric Facility Location}\xspace}
\newcommand{\ikufl}{{\sc Intersection $k$-Nonmetric Facility Location}\xspace}
\newcommand{\ukmfl}{{\sc Union $k$-Metric Facility Location}\xspace}
\newcommand{\ikmfl}{{\sc Intersection $k$-Metric Facility Location}\xspace}
\newcommand{\kds}{{\sc $k$-Densest Subgraph}\xspace}
\newcommand{\leds}{{\sc Minimum $\ell$-Edge Coverage}\xspace}
\newcommand{\mi}{{\sc Matroid Intersection}\xspace}
\newcommand{\pcst}{{\sc Prize-Collecting Steiner Tree}\xspace}

\newcommand{\scthreshold}{\ensuremath{\kappa}}

\newcommand{\OPT}{{\ensuremath{OPT}}}%optimum solution
\newcommand{\opt}{{\ensuremath{opt}}}%its cost

\newcommand{\requests}{\mathcal{U}}
\newcommand{\items}{\mathcal{S}}
\newcommand{\weight}{w}
\newcommand{\approximate}{\mathcal{A}}
\newcommand{\optimal}{\mathcal{O}}

\newcommand{\sciuniv}{\requests}
\newcommand{\scilayerlet}{\items}
\newcommand{\scicost}{\weight}
\newcommand{\scisollet}{\approximate}
\newcommand{\scioptsollet}{\optimal}
\newcommand{\scilayer}[1]{\scilayerlet^{#1}}
\newcommand{\scicostl}[1]{\scicost^{#1}}
\newcommand{\scisol}[1]{\scisollet^{#1}}
\newcommand{\scioptsol}[1]{\scioptsollet^{#1}}
\newcommand{\scioptchoice}[2]{\scioptsol{#1}\left(#2\right)}
\newcommand{\sciK}{K}
\newcommand{\scioptK}{K_{\scioptsollet}}

\newcommand{\facilities}{\mathcal{F}}
\newcommand{\clients}{\mathcal{C}}
\newcommand{\open}{o}
\newcommand{\connect}{w}
\newcommand{\fliconnect}{\connect}
\newcommand{\fliopen}{\open}
\newcommand{\fliuse}{\overline{o}}
\newcommand{\fliuniv}{\clients}
\newcommand{\flilayerlet}{\facilities}
\newcommand{\flisollet}{\approximate}
\newcommand{\flioptsollet}{\optimal}
\newcommand{\flilayer}[1]{\flilayerlet^{#1}}
\newcommand{\fliopenl}[1]{\fliopen^{#1}}
\newcommand{\fliconnectl}[1]{\fliconnect^{#1}}
\newcommand{\flisol}[1]{\flisollet^{#1}}
\newcommand{\flioptsol}[1]{\flioptsollet^{#1}}
\newcommand{\flichoicenolayer}[1]{\flisollet\left(#1\right)}
\newcommand{\flichoice}[2]{\flisol{#1}\left(#2\right)}
\newcommand{\flioptchoice}[2]{\flioptsol{#1}\left(#2\right)}
\newcommand{\fliK}{K}
\newcommand{\flioptK}{K_{\flioptsollet}}

\psset{arrowsize=6pt, labelsep=2pt, linewidth=1.5pt}
\makeatother

\begin{document}

\date{}

\title{Approximation Algorithms for\\ Union and Intersection Covering Problems}

\author{Marek Cygan\inst{1}
\and
Fabrizio Grandoni\inst{2}
\and
Stefano Leonardi\inst{3}
\and
Marcin Mucha\inst{1}
\and
Marcin Pilipczuk\inst{1}
\and
Piotr Sankowski\inst{1}
}
%\begin{titlepage}
%\def\thepage{}
%\thispagestyle{empty}

\institute{Institute of Informatics, University of Warsaw, Poland \\
  \email{\{cygan,mucha,malcin,sank\}@mimuw.edu.pl}
  \and
    University of Rome Tor Vergata, Roma, Italy,
  \email{grandoni@disp.uniroma2.it}
  \and
    Department of Computer and System Science, Sapienza University of Rome, Italy,  \\
  \email{leon@dis.uniroma1.it}
}

\maketitle

\begin{abstract}
\noindent In a classical covering problem, we are given a set of \emph{requests} that we need to satisfy (fully or partially), by buying a subset of \emph{items} at minimum cost. For example, in the $k$-MST problem we want to find the cheapest tree spanning at least $k$ nodes of an edge-weighted graph. Here nodes and edges represent requests and items, respectively.

In this paper, we initiate the study of a new family of \emph{multi-layer} covering problems. Each such problem consists of a collection of $h$ distinct instances of a standard covering problem (\emph{layers}), with the constraint that all layers share the same set of requests. We identify two main subfamilies of these problems:
\begin{itemize}\itemsep0pt
\item[$\bullet$] in a \emph{union} multi-layer problem, a request is satisfied if it is satisfied \emph{in at least one} layer;
\item[$\bullet$] in an \emph{intersection} multi-layer problem, a request is satisfied if it is satisfied \emph{in all} layers.% \emph{simultaneously}.
\end{itemize}
To see some natural applications, consider both generalizations of $k$-MST. Union $k$-MST can model a problem where we are asked to connect at least $k$ users to either one of two communication networks, e.g., a wireless and a wired network. On the other hand, intersection $k$-MST can formalize the problem of providing both electricity and water to at least $k$ users.

We present a number of hardness and approximation results for union and intersection versions of several standard optimization problems: MST, Steiner tree, set cover, facility location, TSP, and their partial covering variants.\end{abstract}
%\end{titlepage}
%\newpage

\section{Introduction}

In the fundamental {\sc Minimum Spanning Tree} problem (\mst), the goal is to compute the cheapest tree which spans all the $n$ nodes of a given edge-weighted graph $G=(V,E)$. To handle the subtleties of real-life applications, several natural generalisations and variants of the problem have been considered. For example, in the \st problem we need to connect with a tree only a given subset $W$ of $k$ \emph{terminal} nodes. In the \kmst problem instead, the goal is to connect at least $k$ (arbitrary) nodes. One common feature of these generalizations is that we need to design a single network.
%One common feature of these generalisations is that we have a single fixed network with a single set of connection costs.
However, this is often not the case in the applications. For example, suppose we want to provide at least $k$ out of $n$ users with both electricity and water. In this case, we cannot design the water and electricity infrastructures independently: our decisions on which users to reach have to be synchronized.

Consider now another classic problem, the {\sc Travelling Salesman} problem (\tsp): here we are given a complete weighted graph, and the goal is to compute the minimum-length tour traversing all the nodes. Again, several natural generalizations and variants of the problem have been considered in the literature. Still, all of them deal only with the case where there is a single network. However, there are natural applications which do not fit in this framework. For example, suppose you want to visit a set of places (bank, post office, etc.), and you can use your bike and your car. Of course, you cannot just reach a place by bike, and then suddenly switch to your car (that you left at home). Your trip must consist of a tour by bike and another tour by car, which together touch all the places that you need to visit.

The above examples show the need for a new framework, which is able to capture coordinated decision-making over multiple optimization problems.

\paragraph{Our results.}

In this paper we initiate the study of \emph{multi-layer covering} problems. These problems are characterized by a set of $h$ instances of a standard covering problem (\emph{layers}), sharing a common set of $n$ \emph{requests}. The goal is \emph{satisfying}, possibly partially, the requests by buying \emph{items} in each layer at minimum total cost. We identify two main families of such problems:
\begin{itemize}\itemsep0pt
\item[$\bullet$] {\bf Intersection problems.} Here, as in the water-electricity example, a request is satisfied if it is satisfied \emph{in all} the layers.
\item[$\bullet$] {\bf Union problems.} Here, as in the car-bike example, a request is satisfied if it is satisfied \emph{in at least} one layer.
\end{itemize}

We provide hardness and approximation results for the union and intersection versions of several classical covering problems: \mst, \st, ({\sc Nonmetric} and {\sc Metric}) \fl, \tsp, and \setc. (Formal definitions are given at the end of this section). We focus on the partial covering variant of these problems, i.e.\ \kmst, \kst, etc.: here we need to satisfy a \emph{target} number $k$ of the $n$ requests. This allows us to handle a wider spectrum of interesting problems. In fact, for intersection problems, if $k=n$ it is sufficient to compute an independent solution for each layer. On the other hand, some of the union problems above are interesting also for the case $k=n$. However, the results that we achieve for that case are qualitatively the same as for $k<n$.

For {\sc Intersection} versions of \kmst, \kst, \ktsp, \ksc, \kmfl, and \kufl, we show that:
\begin{itemize}\itemsep0pt
\item[$\bullet$] Even for two layers, a polylogarithmic approximation for these problems would imply a polylogarithmic approximation for \kds. We recall that the best approximation for the latter problem is $O(n^{\frac{1}{4}+\varepsilon})$ \cite{feige-densestsubgraph-2010} and finding a polylogarithmic approximation is a major open problem. Indeed, many researchers believe that a polylogarithmic approximation does not exist, and exploit this assumption in their hardness reductions (see, e.g.,  \cite{ABW10,ABBG10}).
%The hardness of related planted dense subgraph problems is an assumption in a few recent hardness results \cite{ABW10,ABBG10}.
\item[$\bullet$] On the positive side, we give $\tilde{O}(k^{1-1/h})$-approximation algorithms\footnote{The $\tilde{O}$ notation suppresses polylogarithmic factors.} for these problems.
\end{itemize}
Note that, in the single-layer case, the above problems can be approximated within a constant or logarithmic factor.
Hence, our results show that the complexity of natural intersection problems changes drastically from one to two layers.

For {\sc Union} versions of \kmst, \kst, \ktsp and \kmfl we show that:
\begin{itemize}\itemsep0pt
\item[$\bullet$] The problems are $\Omega(\log k)$-hard to approximate for an unbounded number $h$ of layers. Furthermore, there is a greedy $O(\log k)$-approximation algorithm. For the first three problems this only holds for the rooted version --- the unrooted case is inapproximable.
\item[$\bullet$] There is an LP-based algorithmic framework which provides $O(h)$-approximate solutions. Furthermore, the natural LPs involved have $\Omega(h)$ integrality gap.
\end{itemize}
We remark that \uksc\ and \ukufl\ can be solved by collapsing all layers into one, and hence they are less interesting with respect to the goals of this paper.

\paragraph{Related Work.}
To the best of our knowledge, and somewhat surprisingly, approximation algorithms for union and intersection problems seem to not have been studied in the literature, with the notable exception of \mi. However, differently from our problems, \mi\ is solvable in polynomial time~\cite{edmonds:intersection}.

The term "multi-layer" has been used before in the literature, but with a meaning different from ours. Most often it refers to problems related to VLSI design, where we are given several planar layers on which the circuit has to be built~\cite{multilayer-vlsi}. It also sometimes refers to multi-layer models of communication networks that are composed of different physical and logical layers of communication devices \cite{multilayer-networks}.

The idea of introducing multiple cost functions into one optimization problem
is the main theme of \emph{multi-objective optimization}.
Standard and multi-criteria approximation algorithms have been developed for the multi-objective version of several classical problems, such as {\sc Shortest Path} \cite{hansen,martins,PY00focs,tarapata},
{\sc Spanning Tree} \cite{GRS09esa,GZ10,PY00focs,GR96swat},
{\sc Matching} \cite{BBGS10jcss,BBGS08ipco,GRS09esa,PY00focs}
etc. (for a survey, see~\cite{G00anannotated}). One could view these problems as having several layers with different costs. However, this setting is very different from our approach. In fact, solutions in different layers of multi-objective optimization problems have to be \emph{exactly the same}, and the goal is to satisfy some constraints on each objective.

Partial covering problems (also known as problems with \emph{outliers}), are well-studied in the literature: e.g., \kmst 
\cite{AK06mp,AR98,AABV95,BRV96,Garg96focs,Garg05stoc,RSMRR93},
%\cite{AK06mp,Garg05stoc},
  \ktsp \cite{AK06mp,Garg96focs}, \kmfl \cite{CKMN01soda,JMMSV03}, and \ksc~\cite{kearns,slavik}.
Their generalization on multiple layers is significantly harder, as our results show. Note that our \ukst problem generalizes all of the following problems: \kst (and hence \kmst), \pcst (see the proof of Theorem~\ref{thr:hardnessBounded}), and \ksc (see the proof of Theorem~\ref{thr:hardnessRooted}).

\emph{Rent-or-buy} \cite{EGRS08,FKLS06,GKPR07,SK02} and \emph{buy-at-bulk} \cite{GI06,GR10,GMM01,GKPR07,T02} problems can be seen as multi-layer problems where edge weights in different layers differ by a multiplicative factor. In contrast, weights of different layers are unrelated in our framework. \uktsp has some points in common with \emph{multi-depot} versions of TSP \cite{MRD07,RS10}: also in that case multiple tours are computed; however, their weights are measured w.r.t. a unique weight function. 

Recently Krishnaswamy et al.~\cite{matroid-median} considered a \emph{matroid median} problem,
where a set of open centers must form an independent set from a matroid. 
This can be viewed as a generalisation of a {\sc Union} problem, however 
in the \emph{matroid median} problem all the centers are in the same metric space. This setting is less general than ours as it does not allow for modelling
a setting with several completely unrelated metric spaces.

\paragraph{Preliminaries.}
\label{sec:preliminaries}
In \emph{covering} problems we are given a set $\requests$ of $n$ \emph{requests}, and a set $\items$ of \emph{items}, with costs $w:\items\to \R_{\geq 0}$. The goal is to satisfy all requests by selecting a subset of items at minimum cost.
We already defined \mst, \st, and \tsp. Here, nodes and edges represent requests and items (with costs $w:E\to \R_{\geq 0}$), respectively. In the \setc problem, requests are the elements of a universe $\requests$, and items $\items$ are subsets $S_1,\ldots,S_m$ of $\requests$. Any $S_i$ satisfies all the $v\in S_i$. \ufl is a generalization of \setc, where we are given a set $\facilities$ of \emph{facilities}, with opening costs $\open:\facilities \to \R_{\geq 0}$, and a set $\clients$ of \emph{clients}, with connection costs $\connect:\clients\times \facilities\to \R_{\geq 0}$. The goal is to compute a subset $\flisollet$ of \emph{open} facilities such that $\sum_{f\in \flisollet}\open(f)+\sum_{c\in \clients}\connect(c,\flisollet)$ is minimized. Here $\connect(c,\flisollet):=\min_{f\in \flisollet}\connect(c,f)$. We also say that $c$ is connected to (or served by) $\flichoicenolayer{c}:=\arg\min_{f\in \flisollet}\connect(c,f)$. If connection costs satisfy triangle inequality, the problem is called \mfl.

We can naturally define partial covering versions for the above problems: \kmst, \kst, \ktsp, \kufl, and \kmfl\footnote{In the literature \kufl\ often means that we are allowed to open at most $k$ facilities, while here we mean that we need to connect at least $k$ clients. Similarly for \kmfl. Sometimes \ksc indicates a \setc instance where the largest cardinality of a set is $k$, while our problem is sometimes called {\sc Partial Set Cover}.}.

It is straightforward to define union and intersection versions of  the above problems. In the rest of this paper, the number of layers is denoted by $h$, and variables associated to layer $i$ have an apex $i$ (e.g., $\weight^i$, $\open^i$, etc.), whereas $\OPT$ denotes the optimum solution, and $\opt$ its cost. By $N$ we denote the total number of requests and items (in all layers).

By standard reductions, a $\rho$-approximation for the \kmst problem implies a $2\rho$-approximation
for \kst and \ktsp. Moreover, a $\rho$-approximation for \ktsp gives a $2\rho$-approximation for \kmst. Essentially, the same reductions extend to the union and intersection versions of these problems. For this reason, in the rest of this paper we will consider the union and intersection version of \kmst only. 
%Unless differently stated, we consider the rooted version of \ukmst, where we are given a root $r^i$ per layer which must belong to the solution for that layer.
%The unrooted case is discussed in Appendix \ref{apx:union:unrooted}.

\section{Intersection Problems}\label{s:apx}
\newcommand{\sciproc}{\ensuremath{\mathtt{SCI}}}

%In this section we present our main results on the intersection problems.
%We start with some hardness results.
%Next, we present an algorithm for \iksc, which can be generalized to \ikufl. Based on a different approach, we also give an algorithm for \ikmst.

\subsection{\iksc}

In this section we present our approximation algorithm for \iksc. We recall that in this problem we are given $h$ collections $\scilayer{1}, \scilayer{2},\ldots,\scilayer{h}$ of subsets of a given universe $\sciuniv$, where $\scicostl{i}:\scilayer{i} \to \R_{\geq 0}$ is the cost of subsets in the $i$th collection. The goal is covering at least $k$ elements in all layers simultaneously, at minimum total cost.

The basic idea behind our algorithm is as follows. We consider any set $X$ in any layer, and any number $j\leq k$ of elements in $X$. We solve recursively, on the remaining layers, the intersection problem induced by $X$ with target $j$. The base of the induction is obtained by solving a one-layer \iksc{} problem, using the greedy algorithm which provides a $(1 + \ln k)$-approximation \cite{slavik}. We choose the set $X$ and the cardinality $j$ for which we obtain the best ratio of cost to number of covered elements. Next, we include covered elements in the solution under construction, and the problem is reduced consequently.

In order to highlight the main ideas of our approach, we focus on the special case $h=2$, and we neglect polylogarithmic factors in the analysis. It is easy, just more technical, to extend the same approach to $h>2$ and to refine (slightly) the approximation factor (See Appendix \ref{apx:fli}).
%Our algorithm is decribed in Figure \ref{alg:iksc}.

%We use $N$ to denote the overall number of nodes.

\begin{Figure}
%\begin{algorithm}
\begin{minipage}{\textwidth}
\small
\begin{algorithmic}[1]
  \Procedure{$\sciproc$}{$k, \sciuniv, \scilayer{1}, \scilayer{2}, \scicostl{1}, \scicostl{2}$}
  \State{$\sciK \gets \emptyset$, $\scisol{1} \gets \emptyset$, $\scisol{2} \gets \emptyset$}
  \Repeat
  \For{a=$1$ {\bf to} $2$}
    \For{all $X \in \scilayer{a}$}
      \For{$b := 1$ to $\min(k-|\sciK|, |X \setminus \sciK|)$}
      \State {Solve one-layer \iksc{} problem on layer $\overline{a}$}
      \State {with universe $X\setminus K$ and target $b$.}
%        \State {target $b$. Let $(K_{a,b,X}, \mathcal{A}_{a,b,X})$ be the solution obtained, of cost $C_{a,b,X}$}
      \EndFor
    \EndFor
  \EndFor
%  \State{$a', b', X' \gets$ values of the loops' iterators which minimize $C_{a,b,X} / b$.}
  \State{Let $(a',b',X')$ be the loop iterators which provide a solution $(\sciK',\scisollet')$}
  \State{minimizing the ratio of cost $C'$ to number $b'$ of covered elements.}
  \State{$\sciK \gets \sciK \cup \sciK'$, $\scisol{a'} \gets \scisol{a'} \cup \{X'\}$, $\scisol{\overline{a}'} \gets \scisol{\overline{a}'} \cup \scisollet'$}
  \Until{$|\sciK| = k$}
  \State{\Return $(\sciK, \scisol{1}, \scisol{2})$}
  \EndProcedure
\end{algorithmic}
\end{minipage}
%\end{algorithm}
 \caption{Approximation algorithm for $2$-layer \iksc. For $a\in \{1,2\}$, $\overline{a}$ is the other value in $\{1,2\}$}
\label{alg:iksc}
\end{Figure}

\begin{theorem}\label{thr:intersection:sc}
There is a $\tilde{O}(\sqrt{k})$-approximation algorithm for \iksc on two layers.
\end{theorem}

\begin{proof}
Consider the algorithm in Figure \ref{alg:iksc}. Its running time is polynomial, since $\sciproc$ procedure calls
the one-layer greedy algorithm $O(Nk^2)$ times.

Let $(\scioptsol{1},\scioptsol{2})\subseteq \scilayer{1}\times \scilayer{2}$ be the optimal solution, and let $\scioptK \subseteq (\cup_{S\in \scioptsol{1}} S) \cap (\cup_{S\in \scioptsol{2}}S)$ be any set of $k$ elements in the intersection.
For each element $x \in \scioptK$ and layer $i=1,2$, let us fix a set $\scioptchoice{i}{x} \in \scioptsol{i}$ that covers $x$. 
We prove that at each iteration of the main loop $C' / b' = \tilde{O}(\opt / \sqrt{k-|\sciK|})$.
%  Intuitively we want to show that the average cost of each newly covered element is not greater than the cost of the optimal solution divided by the number of elements still to be covered.
This implies that the total cost of the constructed solution is bounded by
$\sum_{i=0}^{k-1} \tilde{O}(\opt / \sqrt{k-i}) = \opt \cdot \tilde{O}(\sqrt{k}).$

Let $\scthreshold := \sqrt{k-|\sciK|}$. We consider two cases, depending on whether there exists a set $X$ in the optimal solution that covers at least $\scthreshold$ elements of $\scioptK \setminus \sciK$.

{\bf Case 1.} Assume that there exists $1 \le a \le 2$ and $X \in \scioptsol{a}$, such that for at least $\scthreshold$ elements $x$ of $\scioptK \setminus \sciK$ we have $\scioptchoice{a}{x} = X$.
  Let us focus on the moment when our algorithm considers taking the set $X$.
  Obviously we have $\scthreshold \leq k-|\sciK|$, therefore our algorithm considers
  covering $b:=\scthreshold$ elements of $X$.
  As the optimal solution does it, it may be done with cost $\opt$, so the
  call to the one layer algorithm returns a solution with cost $\tilde{O}(\opt)$. Hence we have $C' / b' = \tilde{O}(\opt/ \sqrt{k-|\sciK|})$.

  {\bf Case 2.} For each $1 \le a \le 2$ and every $X \in \scioptsol{a}$,
  at most $\scthreshold$ elements of $\scioptK \setminus \sciK$ satisfy $\scioptchoice{a}{x} = X$.
  For each $x \in \scioptK \setminus \sciK$, let
  $w(x) :=  w^1(\scioptchoice{1}{x}) + w^2(\scioptchoice{2}{x})$
  be the sum of the costs of sets covering $x$ in the optimal solution.
  We have
  $$\sum_{x \in \scioptK \setminus \sciK} w(x) =
  \sum_{a=1}^2 \sum_{X \in \scioptsol{a}} \sum_{x \in \scioptK \setminus \sciK:\ \scioptchoice{a}{x} = X} w^a(\scioptchoice{a}{x}) \leq \sum_{a=1}^2 \sum_{X \in \scioptsol{a}} w^a(X) \scthreshold \leq \scthreshold \cdot \opt.$$
  Thus there exists $x_0 \in \scioptK \setminus \sciK$
  such that $w(x_0) \leq \scthreshold \cdot \opt/|\scioptK \setminus \sciK|$.
  If we take any $a$ and consider the iteration with $X=\scioptchoice{a}{x_0}$
  and $b=1$, the algorithm computes a set of minimum cost $C_0 \leq w(x_0)$
  covering $x_0$. We can conclude that
  $$\frac{C'}{b'}\leq C_0 \leq \frac {\scthreshold \cdot \opt}  {|\scioptK \setminus \sciK|} = \tilde{O}(\opt / {\sqrt{k-|\sciK|}}).$$
%  because $|K^* \setminus K| \geq k-|K|$ and since the call to the one
%  layer problem for $j=1$ returns optimal choice of the covering set.
\qed
\end{proof}

%Algorithm~\ref{alg:approx} can be generalised to arbitrary number of layers $d$.
%The idea is that the algorithm recursively calls itself for one layer less
%until it reaches one layer.
%Furthermore with slight modifications it works also for the \sciname{} problem.
%After carefull calculations and taking treshold
%$P := 4^{1-1/d} (k-|K|)^{1-1/d} \log^{1/d}(k)$
%inside the proof to optimally balance logarithmic factors
%we get the following theorem, which is proven in Appendix.
The proof of the following theorem is in Appendix \ref{apx:fli} due to space limits.
\begin{theorem}\label{thr:intersection:sc2}
There exists a $(4k^{1-1/h}\log^{1/h}(k))$-approximation algorithm for \ikufl (hence for \iksc) running in $N^{O(h)}$ time.
\end{theorem}

\subsection{\ikmst}\label{sec:intersection:kmst}

In this section we present a simple approximation algorithm for \ikmst. 
Recall that here we are given a graph $G=(V,E)$ on $n$ nodes, and $h$ edge-weight functions $w^1,\ldots,w^h$.
By taking the metric closures of $w^i$ we may assume that $G$ is complete.
The goal is computing a tree $T^i$ for each layer such that $\sum_iw^i(T^i)$ is minimized and $|\bigcap_i V(T^i)|\geq k$.

The algorithm is very simple: We consider a new metric $w$ defined as a sum $w(e):=\sum_i w^i(e)$ for each $e\in E$, and compute a $2$-approximate solution of the resulting (one-layer) \kmst\ problem using the algorithm in \cite{Garg05stoc}.

In Appendix~\ref{sec:ommited-kmsti} we prove the following theorem.
\begin{theorem}\label{thm:kmst}
The \ikmst\ algorithm above is $16k^{1-1/h}$-approximate.
%It can be modified to solve \iktsp{} with $32k^{1-1/h}$ approximation guarantee.
\end{theorem}

The analysis of the approximation ratio of the above algorithm given in Theorem~\ref{thm:kmst} is tight up to a factor $O(h)$ (see Appendix~\ref{sec:ommited-kmsti}).

\subsection{Approximation Hardness}\label{s:hard}

This section is devoted to the approximation hardness of \ikmst, \iksc (hence also of \ikufl) and \ikmfl. We use reductions from the \kds problem: find the induced subgraph on $k$ nodes with the largest possible number of edges.
The fact that partial coverage problems can be as hard 
as \kds is already known.
Hajiaghayi and Jain~\cite{hj06} use \kds to show that a partial
coverage version of the {\sc Steiner Forest} problem has no polylogarithmic approximation.
In particular they introduce the \leds problem where one is to find the minimum
number of vertices in a graph, whose induced subgraph has at least $\ell$ edges.
Moreover Hajiaghayi and Jain show a relation between approximation ratios for \kds and \leds.
In order to simplify our reductions we extend the result on \leds to bipartite graphs
and prove the following theorems in Appendix~\ref{apx:kmst-hardness}.
%The proof of the following theorem is given in Appendix \ref{apx:kmst-hardness}.

%\begin{theorem}\label{thm:dense-v2e}
%If there exists an $f(n)$-approximation algorithm for \leds on bipartite graphs, then there exists a $16(f(2n))^2$-approximation algorithm for \kds on arbitrary graphs.
%\end{theorem}

%Using Theorem~\ref{thm:dense-v2e} together with simple reductions we obtain the following theorems proved in Appendix~\ref{apx:kmst-hardness}.

\begin{theorem}\label{thm:dse-sci}
If there exists an $f(n)$-approximation algorithm for unweighted \iksc on two layers or for \ikmfl on two layers, then there exists a $16(f(2m))^2$-approximation algorithm for \kds.
\end{theorem}

\begin{theorem}\label{thm:dse-mst}
If there exists an $f(n)$-approximation algorithm for \ikmst on two layers, then there exists a $16(f(2n+2m+2))^2$-approximation algorithm for \kds.
\end{theorem}

Theorems \ref{thm:dse-sci} and \ref{thm:dse-mst}
suggest that the existence of a polylogarithmic approximation algorithm
for the considered problems is rather unlikely (or at least very hard to achieve).

\section{Union Problems}

In this section we present our results for \ukmst and \ukmfl. For \ukmst, we first consider the rooted case in Sections \ref{sec:union:hardness} and \ref{sec:union:lp} and then the unrooted one in Section \ref{union:unrooted}. The \umst\ problem is a variant of \ukmst\ with $k=n$.

\subsection{Approximation Hardness}\label{sec:union:hardness}

\begin{theorem}\label{thr:hardnessBounded}
Rooted \ukmst and \ukmfl are APX-hard for any $h\geq 1$. \umst is APX-hard for any $h\geq 2$.
\end{theorem}
\begin{proof}
The first claim trivially follows from the APX-hardness \cite{Garg05stoc,GK98soda} of the considered problems for $h=1$, by adding dummy layers with infinite edge weights.

%The first claim follows from the APX-hardness of classical rooted $k$-MST, by introducing $(h-1)$ dummy layers with infinite edge weights. This reduction preserves approximation factors.
%For the spanning tree case,
For the second claim, we consider a reduction from the APX-hard \cite{BP89ipl} \pcst\ problem: given an undirected graph $G=(V,E)$, edge weights $w:E\to \R_{\geq 0}$, a root node $r\in V$, and node prizes $p: V\to \R_{\geq 0}$, find a tree $T\ni r$ which minimizes $\sum_{e\in T}w(e)+\sum_{v\notin T}p(v)$. We create a first layer, with edge weights $w^1=w$. Then we construct a second layer, where we set $w^2(\{r,v\})=p(v)$ for any $v\in V$. All the other layers, if any, are dummy layers defined as above. This reduction is approximation preserving.
\qed
\end{proof}

For an unbounded number of layers, our problems become much harder.
\begin{theorem}\label{thr:hardnessRooted}
For an arbitrary number of layers, rooted \ukmst and \ukmfl are not approximable better than $\Omega(\log k)$ unless $P=NP$, even when $k=n$.
\end{theorem}
\begin{proof}
We prove the claim for rooted \ukmst, by giving a reduction from cardinality \setc: given a universe $\requests$ of $n'$ elements, and a collection $\items=\{S_1,\ldots,S_{m'}\}$ of $m'$ subsets of $\requests$, find a minimum cardinality subset $\mathcal{A}\subseteq \items$ which spans $\requests$. This problem is $\Omega(\log n')$-hard to approximate \cite{RS97stoc}. We create one node per element of $\requests$, plus two extra nodes $r$ and $s$. We create one layer $i$ for each set $S_i$ (i.e., $h=m'$). In layer $i$ we let $w^i(\{r,s\})=1$ and $w^i(\{s,v\})=0$ for each $v\in S_i$. We also let $r^i:=r$ for each $i$, and assume $k=n=n'+2$. Note that any solution to the rooted \ukmst instance of cost $\alpha$ can be turned into a solution to the \setc instance of the same cost, and vice versa.

To prove the claim for \ukmfl, we use the same reduction as above, where the edge $\{r,s\}$ is replaced by a single node $r$, which is a facility of opening cost $1$.
\qed
\end{proof}
In Appendix \ref{apx:union:greedy} we give a greedy $O(\log k)$-approximation algorithm.

\subsection{An LP-Based Approximation for rooted \ukmst} \label{sec:union:lp}

In this section we present an LP-based $O(h)$-approximation algorithm for rooted \ukmst. Essentially the same approach works also for \ukmfl (see Appendix \ref{apx:union:extensions}). This is an improvement over the $\Theta(\log k)$-approximation given by the greedy algorithm for the relevant case of bounded $h$.

%The basic idea is as follows. We consider an LP relaxation for the problem, where, for each request $v$, there is a variable $z^i_v$ indicating whether $v$ is satisfied in layer $i$. Variable $z_v$, $\sum_i z^i_v\geq z_v$, indicates whether $v$ is satisfied in at least one layer. We solve the LP, and use the optimal solution to partition the requests in (disjoint) groups $W^1,\ldots,W^h$: request $v$ is assigned greedily to the layer $i$ with the largest $z^i_v$. Then we solve independently the single-layer problem induced by each $W^i$, with target roughly equal to the sum of $z_v$ over $v\in W^i$. The idea is computing on each layer solutions of cost bounded with respect to the value of a proper LP (related to the original LP).

%For ease of presentation, we next focus on \ukmst. A $O(h)$-approximation for the other problems is given in Appendix \ref{apx:union:extensions}.
%Here, we focus on \ukmst\ (see Appendix \ref{apx:union:extensions} for related results).
For notational convenience, we assume that the roots $R:=\cup_{i}\{r^i\}$ are not counted into the target number $k$ of connected nodes. In other terms, we replace $k$ by $k-|R|$. We make the same assumption also in the case of one layer. 
%\kst\ is a generalization of \kmst, where one wishes to connect $k$ out of a subset $W\subseteq V-\{r\}$ of terminals to the root $r$. In particular, for $W=V$, \kst\ reduces to \kmst.
Consider the following LP relaxation for \kst ($W\ni r$ is the set of terminals)
denoted by $LP_{kST}(w,W,V,r,k)$:
\begin{align*}
\min\quad & \textstyle{\sum_{e\in E}w(e)\,x_{e}} &  &\\
 s.t.\quad  & \textstyle{\sum_{e\in \delta(S)}x_{e}\geq z_{v},} & \textstyle{\forall (v,S): S\subseteq V-\{r\},v\in S\cap W;} & \\
        & \textstyle{\sum_{v\in W} z_v\geq k;} & & \\
        & \textstyle{x_{e}\geq 0,1\geq z_v\geq 0,} & \textstyle{\forall v\in W,\forall e\in E.} &
\end{align*}
Here, variable $x_e$ indicates whether edge $e$ is included in the solution, whereas variable $z_v$ indicates whether terminal $v$ is connected. Moreover $\delta(S)$ denotes the set of edges with exactly one endpoint in $S$. 
Observe that 
$LP_{kMST}(w,V,r,k):=LP_{kST}(w,V,V,r,k)$ is an LP relaxation for \kmst. We need the following lemmas.
\begin{lemma}\label{lem:garg} \cite{Garg96focs}
Let $(w,V,r,k)$ be an instance of \kmst, $w_{max}:=\max_{v\in V}w(r,v)$, and $\opt'$ be the optimal solution to
$LP_{kMST}(w,V,r,k)$. There is a polynomial-time algorithm {\tt apx-kmst} which computes a solution to the instance of cost at most $2\opt'+w_{max}$.
\end{lemma}
%The following lemma is implicit in \cite{Frank89adm,Jackson88jgt}.
\begin{lemma}\label{lem:frank}\cite{Frank89adm,Jackson88jgt}
Let $G=(V\cup \{v\},E)$ be a directed graph, with edge capacities $\alpha:E\to \R_{\geq 0}$ such that $\sum_{e\in \delta^+(u)}\alpha(e)=\sum_{e\in \delta^-(u)}\alpha(e)$ for all $u\in V\cup \{v\}$. Then there is a pair of edges $(u,v)$ and $(v,z)$, such that the following capacity reservation $\beta$ supports the same flow as $\alpha$ between any pair of nodes in $V$: for $\Delta\alpha:=\min\{\alpha(u,v),\alpha(v,z)\}$, set $\beta(u,v)=\alpha(u,v)-\Delta\alpha$, $\beta(v,z)=\alpha(v,z)-\Delta\alpha$, $\beta(u,z)=\alpha(u,z)+\Delta\alpha$, and $\beta(e)=\alpha(e)$ for the remaining edges $e$.
%, and metric edge-weights $w:E\to \R_{\geq 0}$. Then there is a polynomial-time algorithm which, for any given node $v$, computes a new capacity reservation $\beta$ such that: (a) the directed flow supported by $\beta$ from $a$ to $b$, $a,b\in V\setminus \{v\}$, is at least the same as for $\alpha$, (b) $\sum_{e\in E}w(e)\beta(e)\leq \sum_{e\in E}w(e)\alpha(e)$, and (c) $\beta(e)=0$ for the edges incident to $v$.
\end{lemma}

\begin{corollary}\label{cor:frank}
Given a feasible solution $(x,z)$ to $LP_{kST}(w,W,V,r,k)$, there is a feasible solution $(x',z')$ to $LP_{kMST}(w,W,r,k)$ such that $\sum_{e}w(e)x'_e\leq 2\cdot \sum_{e}w(e)x_e$.
\end{corollary}
\begin{proof}
Variables $x_e$ can be interpreted as a capacity reservation which supports a fractional flow of value $z_v$ from each $v\in W$ to the root. Let us replace each edge with two oppositely directed edges, and assign to each such edge the same weight and capacity as the original edge. This way, we obtain a capacity reservation $\alpha$ which costs twice the original capacity reservation, and satisfies the condition of Lemma \ref{lem:frank}. We consider any non-terminal node $v\neq r$ with some incident edge of positive capacity, and apply Lemma \ref{lem:frank} to it.
Due to triangle inequality, the cost of the capacity reservation does not increase. We iterate the process on the resulting capacity reservation. Within a polynomial number of steps, we obtain a capacity reservation $\beta$ which: (1) supports the same flow from each terminal to the root $r$ as $\alpha$, (2) has value $0$ on edges incident to non-terminal nodes (besides $r$), and (3) does not cost more than $\alpha$. At this point, we remove the nodes $V-(W\cup \{r\})$, and merge the capacity of oppositely directed edges to get an undirected capacity reservation $x'$. By construction, the pair $(x',z)$ is a feasible solution to $LP_{kMST}(w,W,r,k)$ of cost at most $2\cdot \sum_{e}w(e)x_e$.
\qed
\end{proof}

We are now ready to describe our algorithm for rooted \ukmst. In a preliminary step we guess the largest distance $L$ in the optimal solution between any connected node and the corresponding root, and discard nodes at distance larger than $L$ from their root. This introduces a factor $O(nh)$ in the running time. Note that $L\leq \opt$. We let $V^i$ be the remaining nodes in layer $i$. 

Then we compute the optimal fractional solution $OPT^*=(x^i,z^i,z)_{i}$, of cost $opt^*$, to the following LP relaxation $LP_{ukMST}$ for the problem,
where variables $x^i_e$ and $z^i_v$ indicate whether edge $e$ is included in the solution of layer $i$ and node $v$ is connected in layer $i$, respectively. Variable $z_v$ indicates whether node $v$ is connected in at least one layer.

\begin{align*}
\min\quad & \textstyle{\sum_{i=1,\ldots,h}\sum_{e\in E}w^i(e)\,x^i_{e}} &  &\\
 s.t.\quad  & \textstyle{\sum_{e\in \delta(S)}x^i_{e}\geq z^i_{v},} & \textstyle{\forall i\in \{1,\ldots,h\}, \forall (v,S):S\subseteq V^i-\{r^i\},v\in S;} & \\
        & \textstyle{\sum_{i=1,\ldots,h} z^i_{v}\geq z_v,} & \textstyle{
\forall v\in V-R;} & \\
        & \textstyle{\sum_{v\in V-R} z_v\geq k;} & & \\
        & \textstyle{z^i_{v},x^i_{e}\geq 0,1\geq z_v\geq 0,} & \textstyle{
\forall i\in \{1,\ldots,h\},\forall v\in V-R,\forall e\in E.} &
\end{align*}

Given $OPT^*$, we compute for each layer $i$ a subset of nodes $W^i$, where $v$ belongs to $W^i$ iff $z^i_v=\max_{j=1,\ldots,h}\{z^j_v\}$ (breaking ties arbitrarily). We also define $k^i:=\lfloor \sum_{v\in W^i}z_v \rfloor$. For each layer $i$, we consider the $k$-MST instance on nodes $W^i\cup \{r^i\}$ with target $k^i$. This instance is solved using the $2$-approximation algorithm {\tt apx-kmst} of Lemma \ref{lem:garg}: the resulting tree $T^i$ is added to the solution for layer $i$. Let $k'$ be the number of connected nodes. If $k'<k$, the algorithm connects $k-k'$ extra nodes, chosen greedily, to the corresponding root in order to reach the global target $k$.
\begin{theorem}\label{thr:constantkMST}
There is a $O(h)$-approximation algorithm for rooted \ukmst. The running time of the algorithm is $O((nh)^{O(1)})$.
\end{theorem}
\begin{proof}
Consider the algorithm above. The claim on the running time is trivial. By construction, the solution computed is feasible (i.e., it connects $k$ nodes). It remains to consider the approximation factor.

For each $v\in W^i$, we let $\tilde{z}^i_v=z_v$, and set $\tilde{z}^i_v=0$ for the remaining nodes. Furthermore, we let $\tilde{x}^i_e=h\cdot x^i_e$. Observe that $(\tilde{x}^i,\tilde{z}^i,z)_{i}$ is a feasible fractional solution to $LP_{ukMST}$ of cost $h\cdot opt^*$. Observe also that $(\tilde{x}^i,\tilde{z}^i)$ is a feasible solution to $LP_{kST}(w^i,W^i,V^i,r^i,k^i)$: let $\tilde{apx}^i$ be the associated cost. By Lemma \ref{lem:frank}, there is a fractional solution to $LP_{kMST}(w^i,W^i,r^i,k^i)$ of cost at most $2\tilde{apx}^i$. It follows from Lemma \ref{lem:garg} that the solution computed by {\tt apx-kmst} on layer $i$ costs at most $4\tilde{apx}^i+L$.

Since the $W^i$'s are disjoint, the algorithm initially connects at least $\sum_{i}k^i\geq k-h$ nodes. Hence the cost of the final augmentation phase is at most $h\cdot L\leq h\cdot opt$. Putting everything together, the cost of the solution returned by the algorithm is at most:
$$
\sum_{i}(4\cdot\tilde{apx}^i+L)+h\cdot L\leq 4h\cdot opt^*+2h\cdot L\leq 6h\cdot opt\,. \qquad \qed
$$
\end{proof}

The constant multiplying $h$ in the approximation factor can be reduced with a more technical analysis, at the cost of a higher running time. We also observe that the integrality gap of $LP_{ukMST}$ is $\Omega(h)$ (see Appendix \ref{apx:union:gap}).

\subsection{Unrooted \ukmst}
\label{union:unrooted}

\begin{theorem}\label{thr:inapproximable}
Unrooted \ukmst\ is not approximable in polynomial time for an arbitrary number $h$ of layers unless $P=NP$.
\end{theorem}
\begin{proof}
We give a reduction from SAT: given a CNF boolean formula on $m'$ clauses and $n'$ variables, determine whether it is satisfiable or not. For each variable $i$, we create two nodes $t_i$ and $f_i$. Intuitively, these nodes represent the fact that $i$ is true or false, respectively. Furthermore, we have a node for each clause. Hence the overall number of nodes is $n=2n'+m'$. We create a separate layer for each variable $i$ (i.e., $h=n'$). In layer $i$, we connect with an edge of cost zero $t_i$ (resp., $f_i$) to all the clauses which are satisfied by setting $i$ to true (resp., to false)\footnote{Without loss of generality, we can assume that each clause does not contain both a literal and its negation.}. The target value is $k=n'+m'$. Note that, there is a satisfying assignment to the SAT instance iff there is a solution of cost zero to the \ukmst\ instance.
\qed
\end{proof}

For $h=O(1)$, the rooted and unrooted versions of the problem are equivalent approximation-wise. In fact, one obtains an approximation-preserving reduction from the unrooted to the rooted case by guessing one node $r^i$ in the optimal solution per layer: this introduces a polynomial factor $O(n^h)$ in the running time. We remark that an exponential dependence on $h$ of the running time is unavoidable in the unrooted case, due to Theorem \ref{thr:inapproximable}. An opposite reduction is obtained by appending $n$ dummy nodes to each root (distinct nodes for distinct layers), with edges of cost zero, and setting the target to $k+hn$. The following result follows.
\begin{corollary}
Unrooted \ukmst  is APX-hard for any $h\geq 1$. There is a $O(h)$-approximation algorithm for the problem of running time $O((hn)^{O(1)}n^{h})$.
\end{corollary}
%\begin{corollary}
%There is a $O(h)$-approximation algorithm for unrooted \ukmst. The running time of the algorithm is $O((hn)^{O(1)}n^{h})$.
%\end{corollary}
%\begin{proof}
%With an extra factor $O(n^{h})$ in the running time we can reduce the unrooted case to the rooted one. The claim follows from Theorem \ref{thr:constantkMST}.
%\end{proof}

\section{Conclusions and Open Problems}

In this paper, we introduced multi-layer covering problems, a new framework that can be used to describe a wide spectrum of yet unstudied problems. We addressed two natural ways of combining
the layers: intersection and union.
We gave multi-layer approximation algorithms, as well as hardness results, for 
a few classic covering problems (and their partial covering versions).
There are several research questions that merit further study.
%the following classic covering problems (and their partial covering versions): \tsp, \mst, \st, \setc, \mfl and \ufl. There are several research questions that merit further study.
\begin{itemize}
\item There are other natural ways one can combine the layers. Consider, for example, the car/bike problem in the case where you can put your bike in the car trunk. Now you can make more than one tour by bike, the only requirement being that the bike tours all touch the (unique) car tour. 
\item What about min-max multi-layer problems, where the goal is minimizing the maximum cost over the layers?
\item We considered covering problems: what about packing problems?
\item Our algorithms for union problems give tight bounds only with respect to the corresponding natural LPs. This leaves room for improvement. 
\item There is a considerable gap between upper and lower bounds for intersection problems. In particular, our hardness results do not depend on $h$, while the approximation ratios deteriorate rather rapidly for increasing $h$.
\end{itemize}

\bibliographystyle{plain}
\bibliography{literature}
%\bibliography{set-cover-intersection,related}

\appendix
\newpage

\section{Intersection Problems}
\label{apx:intersection}

In this section we give the omitted details about the intersection problems.

\subsection{\ikmst}
\label{sec:ommited-kmsti}

In this section we prove a simple approximation algorithm for \ikmst. 
Recall that here we are given a graph $G=(V,E)$ on $n$ nodes, and $h$ edge-weight functions $w^1,\ldots,w^h$.
By taking the metric closures of $w^i$ we may assume that $G$ is complete.
The goal is computing a tree $T^i$ for each layer such that $\sum_iw^i(T^i)$ is minimized and $|\bigcap_i V(T^i)|\geq k$.

The algorithm is very simple: We consider a new metric $w$ defined as a sum $w(e):=\sum_i w^i(e)$ for each $e\in E$, and compute a $2$-approximate solution of the resulting (one-layer) \kmst\ problem using the algorithm in \cite{Garg05stoc}.

\begin{lemma}\label{lem:kmst2}
Let $K\subseteq V$, and $w^i(K)$ denote the cost of the minimum spanning tree of $K$ on layer $i$. Then there exist two nodes $u,v\in K$ such that $w^i(u,v) \leq 4w^i(K)/|K|^{1/h}$ for $i=1,\ldots,h$.
\end{lemma}
\begin{proof}
Let us prove the following claim by induction on $i$: for any $i \in \{0,\ldots,h-1\}$, there exist a nodeset $K_i \subseteq K$ and paths $P_i^1,P_i^2\ldots,P_i^i$ on $K_i$ such that: (a) $|K_i|\geq |K|^{1-i/h}$ and (b) $w^j(P_i^j)\leq 2w^j(K)/|K|^{1/h}$ for $j=1,\ldots,i$. Trivially $K_0=K$ satisfies the claim, hence assume $i>0$. Let $T^i$ be the minimum spanning tree of $K$ on layer $i$. Duplicate its edges, compute an Euler tour, and shortcut duplicated nodes. Let $C^i$ be the resulting cycle on $K$ of length at most $2w^{i}(K)$. Remove up to $|K|^{1/h}$ edges from $C^i$  so as to obtain $|K|^{1/h}$ segments of length at most $2w^{i}(K)/|K|^{1/h}$ each. Let $P$ be the segment maximizing the cardinality of $K_i:=V(P)\cap K_{i-1}$.
Set $K_i$ satisfies (a) since $|K_i|\geq |K_{i-1}|/|K|^{1/h}\geq |K|^{1-(i-1)/h-1/h}$. The paths $P_i^i$ and $P_i^j$, $j<i$, satisfying (b) are obtained from $P$ and $P_{i-1}^j$, respectively, by shortcutting the nodes not in $K_i$.

Similarly as above, we can split $C^h$ into $|K|^{1/h}/2$ segments which span $K$ and have length at most $4w^{h}(K)/|K|^{1/h}$ each. At least one of these segments contains $2|K_{h-1}|/|K|^{1/h}\geq 2$ nodes of $K_{h-1}$. Thus there are two nodes $u$ and $v$ such that $w^i(u,v)\leq 4w^i(K)/|K|^{1/h}$ for $i=1,\ldots,h$.
\qed
\end{proof}

\begin{theorem}[Theorem~\ref{thm:kmst} restated]
The \ikmst\ algorithm above is $16k^{1-1/h}$-approximate.
%It can be modified to solve \iktsp{} with $32k^{1-1/h}$ approximation guarantee.
\end{theorem}
\begin{proof}
Consider the following process: starting with the optimal set $K_{\optimal}$ of $k$ covered nodes, we iteratively take the edge $\{x,y\}$ guaranteed by Lemma \ref{lem:kmst2} and contract it in all layers, until $K_{\optimal}$ collapses into a single node. The contracted edges form a tree $T'$ (same for all layers) spanning $k$ nodes, of cost
\begin{align*}
 w(T') \leq &\ 4\sum_{i=1}^h w^i(K_{\optimal}) \sum_{i=1}^{k-1} (k-i+1)^{-\frac1h}
\leq 4\sum_{i=1}^h w^i(K_{\optimal}) \int_1^k x^{-\frac1h} dx \\
          <&\ 8k^{1-\frac1h}\sum_{i=1}^h w^i(K_{\optimal})=8k^{1-\frac1h} \opt.
\end{align*}
The algorithm returns a solution of cost at most $2w(T')$. The claim follows.
%For \iktsp{}, recall that both considered problems are equivalent up to a factor of $2$
%in the approximation.
\qed
\end{proof}

We show that the bound proven in Section \ref{sec:intersection:kmst} is tight up to a factor $O(h)$. 
\begin{lemma}
The approximation ratio of the \ikmst\ algorithm in Section \ref{sec:intersection:kmst} is $\Omega(\frac{1}{h}n^{1-1/h})$
\end{lemma}
\begin{proof}
Take an arbitrary integer $N > 2$ and set $n = 2^{hN} - 1$. We are going to
construct weights $w^1, w^2, \ldots, w^h$ on an $n$-node complete graph $G=(V,E)$
such that $w^i(T^i) = n-1$, but $w(T) = \Omega(n^{2-1/h})$. Here $T^i$ is the minimum spanning tree on layer $i$. 

We take $V = \{0, 1,\ldots, n-1\}$, that is, the set of nodes of $G$ are all
numbers with up to $hN$ digits in binary, except for $2^{hN}-1$, i.e., the number with $hN$ ones in binary.
Given $x \in V$, we split its $hN$-digit binary representation into $h$ segments of length $N$ and
denote the $i$-th segment by $x_i$. In other words, the binary representation of $x$ is $x_1x_2\ldots x_h$
and each $x_i$ is a $N$-digit binary string. By $x(i)$ we denote number represented in binary as $x_ix_{i+1}\ldots x_hx_1\ldots x_{i-1}$,
that is, we rotate cyclically $(i-1)N$ bits.

To construct metric $w^i$, sort $V$ according to numbers $x(i)$ and connect $V$ into Hamiltonian cycle $C_i$
in this order. All edges on $C_i$ have weight $1$ and other distances are minimum length distances on $C_i$.
Clearly, $w^i(T^i) = n-1$.

It is sufficient to show that, for each edge ${x,y}$, we have $w(x,y) = \Omega(n^{1-1/h})$. This leads to the claimed bound on $w(T)$.
%Let $x, y \in V$ and $x \neq y$. 
We distinguish a few subcases.

{\bf{Case 1.}} There exists $i$, $1 \leq i \leq h$, such that $|x_i - y_i| \geq 2$ and
$\{x_i, y_i\} \neq \{0, 2^N-1\}$. Then in $G_i$ the distance between $x$ and $y$ is at least
$n^{1-1/h}$.

{\bf{Case 2.}} There exists $i$, $1 \leq i \leq h$, such that $x_i = y_i$. Take $j$ such that $x_j \neq y_j$ but $x_{j+1} = y_{j+1}$ (with $x_{h+1} = x_1$). Then in $G_j$ the distance
between $x$ and $y$ is at least $n^{1-1/h} - n^{1-2/h}$.

{\bf{Case 3.}} There exists $i$, $1 \leq i \leq h$, such that $|x_i - y_i| = 1$.
Then in $G_{i-1}$ (with $G_0 = G_h$) the distance between $x$ and $y$ is at least $n^{1-1/h} - 2n^{1-2/h}$.

{\bf{Case 4.}} As $2^{hN}-1$ is not in $V$, there exists $i$, $1 \leq i \leq h$, such that
$x_i = 2^N-1$, $y_i = 0$, but $x_{i+1} = 0$ and $y_{i+1} = 2^N-1$. Then the distance
between $x$ and $y$ in $G_i$ is at least $n^{1-1/h} - 2n^{1-2/h}$.

As we exhausted all possibilities, the bound on $w(T)$ is proven.
\qed
\end{proof}

\subsection{Approximation hardness}\label{apx:kmst-hardness}

We start with the following technical lemma.
\begin{lemma}\label{extract-div-c}
  Assume we have an undirected graph $G=(V,E)$
  and an induced subgraph $G[X]$, $X \subseteq V$ with $x$ nodes
  and $y$ edges. Let $2 \leq x_0 \leq x$ and let $c := x/x_0$.
  Then one can in polynomial time find an induced subgraph $G[Y]$
  on $x_0$ nodes with at least $y/(2c^2)$ edges.
\end{lemma}

\begin{proof}
  By the linearity of expectations a random 
  subset $X_0 \subseteq X$ containing $x_0$ vertices
  induces a subgraph $G[X_0]$ with $\frac{yx_0(x_0-1)}{x(x-1)} \ge \frac{yx_0^2}{2x^2} = \frac{y}{2c^2}$
  edges.
  We can derandomize this procedure with standard techniques.
\qed
\end{proof}

%\begin{proof}
%  Construct $G[Y]$ greedily: first set $Y_x := X$ and
%  construct $Y_{k-1}$ by removing the lowest degree node from $G[Y_k]$.
%  Take $Y := Y_{x_0}$.
%  We now prove that such construction yields a subgraph with at least $y/(2c^2)$ edges.
%
%  Let $y_k$ be the number of edges in graph $G[Y_k]$. We have $y_x = y$.
%  At each step we have a graph with $k$ nodes and $y_k$ edges, thus
%  by removing a node with the smallest degree we remove at most $2y_k/k$ edges.
%  Thus:
%  $$\frac{y_{k-1}}{k-1} \geq \frac{y_k - \frac{2y_k}{k}}{k-1} = \frac{y_k}{k} \cdot \frac{k-2}{k-1}.$$
%  Hence we have:
%  $$
%    \frac{y_{x_0}}{x_0} \geq \frac{y}{x} \cdot \prod_{k=x_0+1}^x \frac{k-2}{k-1}
%    = \frac{y}{x} \cdot \frac{x_0-1}{x-1}
%    \geq \frac{y}{x} \cdot \frac{1}{2c}.$$
% We can conclude that
%  $$y_{x_0} \geq y \cdot \frac{x_0}{x} \cdot \frac{1}{2c} = y\cdot \frac{1}{2c^2}.$$
%\end{proof}

Now we reduce the domain of \kds\ to bipartite graphs (see also Figure \ref{fig:lemma-ds-bipartite}).
%\rem{F: an example of the construction for N=3 and h=2 would help}
\begin{lemma}\label{ds-bipartite}
  Assume there exists a $f(n, k)$-approximation algorithm for \kds\ on bipartite graphs. Then there exists a $8f(2n, 2k)$-approximation
  algorithm for the same problem on arbitrary graphs.
\end{lemma}
%\rem{F: check that we use (u,v) not uv for edges}
\begin{proof}
  Assume we have an instance of \kds, i.e., a graph $G=(V,E)$ and one integer $k$.
  Construct a bipartite graph $G'=(V_1\cup V_2, E')$ as follows: for each $v \in V$ we take two
  copies $v_1 \in V_1$ and $v_2 \in V_2$. For each $(u,v) \in E$ we add $(u_1,v_2)$ and
  $(u_2,v_1)$ to $E'$. Let us run the $f(n,k)$-approximation algorithm for the bipartite graph $G'$
  and $k' := 2k$. This way we obtain a set $X' \subset V_1 \cup V_2$ such that $G[X']$ has $y'$ edges.
  Take $X := \{v: v_1 \in X' {\rm\ or\ } v_2 \in X'\}$, $k \leq |X| \leq 2k$. The graph $G[X]$ has at least $y'/2$ edges. Reduce $X$ to size $k$ using Lemma \ref{extract-div-c},
  obtaining a solution with at least $y'/16$ edges.

Let us now bound how much the obtained solution is worse than the optimal solution. Let
  $X_{\opt}$ be any optimal solution in $G$ such that $G[X_{\opt}]$ has $y_{\opt}$ edges
  and $k$ nodes. In $G'$ set $X_{\opt}' = \{v_1, v_2: v \in X_{\opt}\}$ has
  $2k$ nodes and $2y_{\opt}$ edges, thus $y' \geq 2y_{\opt}/f(2n, 2k)$. Therefore
  $G[X]$ has at least $y_{\opt}/(8f(2n, 2k))$ edges.
\qed
\end{proof}

\begin{Figure}[t]
\begin{center}
\begin{tikzpicture}[scale=1.0]
  \draw[thick] (0.0, -1.0) -- (0.0, 1.0) -- (-0.9, 2.0) -- (-0.9, 0.0) -- (0.0, -1.0);
  \fill (0.0, -1.0) circle (0.1);
  \fill (0.0, 1.0) circle (0.1);
  \fill (-0.9, 2.0) circle (0.1);
  \fill (-0.9, 0.0) circle (0.1);

  \draw[dashed,->] (1.0, 0.5) -- (3.0, 0.5);
  \begin{scope}[shift={(4.0, 0.0)}]
    \fill (0.0, -1.0) circle (0.1);
    \fill (0.0, 1.0) circle (0.1);
    \fill (-0.0, 2.0) circle (0.1);
    \fill (-0.0, 0.0) circle (0.1);
  \draw[thick] (0.0, -1.0) -- (1.0, 0.0);
  \draw[thick] (0.0, -1.0) -- (1.0, 1.0);
  \draw[thick] (0.0, 1.0) -- (1.0, 2.0);
  \draw[thick] (0.0, 1.0) -- (1.0, -1.0);
  \draw[thick] (-0.0, 0.0) -- (1.0, -1.0);
  \draw[thick] (-0.0, 0.0) -- (1.0, 2.0);
  \draw[thick] (-0.0, 2.0) -- (1.0, 0.0);
  \draw[thick] (-0.0, 2.0) -- (1.0, 1.0);
  \end{scope}
  
  \begin{scope}[shift={(5.0, 0.0)}]
    \fill (0.0, -1.0) circle (0.1);
    \fill (0.0, 1.0) circle (0.1);
    \fill (0.0, 2.0) circle (0.1);
    \fill (0.0, 0.0) circle (0.1);
  \end{scope}
\end{tikzpicture}
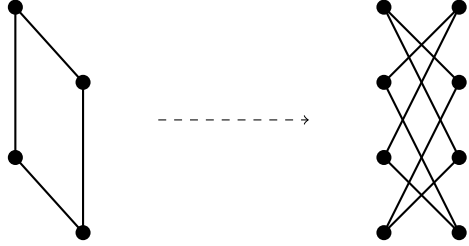
\caption{Construction from Lemma \ref{ds-bipartite} applied to a $4$-cycle.}
\label{fig:lemma-ds-bipartite}
\end{center}
\end{Figure}

Now we relate \kds\ and \leds.
A lemma similar to the following one was proved in~\cite{hj06}, but we include the proof for the sake of completeness.

\begin{lemma}\label{ds-dse}
  Assume there exists a $f(n)$-approximation algorithm
  for \leds\ on bipartite graphs.
  Then there exists a $2(f(n))^2$-approximation algorithm
  for \kds\ on bipartite graphs.
\end{lemma}

\begin{proof}
  Assume we have an instance $(G=(V,E), k)$ of \kds.
  For the graph $G$ we run the approximation algorithm for \leds\
  with consecutive $\ell:=1,2,\ldots$, obtaining solutions $X_1, X_2, \ldots$.
  We stop when $G[X_{\ell+1}]$ has more that $f(n)k$ nodes. Assume that
  $X_\ell$ was the last solution with at most $f(n)k$ nodes.
  We reduce $X_\ell$ to size $k$ using Lemma \ref{extract-div-c} and return the reduced set.

  Let us now prove that it is in fact a $2(f(n))^2$-approximation.
  Let $X_{\opt}$ be the optimal solution for the \kds\ instance with $y_{\opt}$ edges.
  Note that $\ell \geq y_{\opt}$, as $X_\ell$ was the last solution with at most $f(n)k$ nodes
  and our algorithm is a $f(n)$-approximation. Thus, by Lemma \ref{extract-div-c},
  the returned solution has at least $\ell/(2(|X_\ell|/k)^2) \geq y_{\opt}/(2(f(n))^2)$ edges.
\qed
\end{proof}

%\begin{remark}
%  The same thesis as in Lemma \ref{ds-dse} holds for \leds\
%  and \kds\ considered in arbitrary graphs. However, we stated Lemma \ref{ds-dse}
%  for bipartite graphs, to pipeline it with Lemma \ref{ds-bipartite}.
%\end{remark}

Pipelining Lemmas~\ref{ds-dse}~and~\ref{ds-bipartite} proves the following theorem.

\begin{theorem}\label{thm:dense-v2e}
If there exists an $f(n)$-approximation algorithm for \leds on bipartite graphs, then there exists a $16(f(2n))^2$-approximation algorithm for \kds on arbitrary graphs.
\end{theorem}

We conclude this section by showing the missing reductions for \iksc, \ikmst and \ikmfl.

\begin{lemma}\label{lem:dse-sci}
If there exists an $f(n, k)$-approximation algorithm for unweighted \iksc on two layers, then there exists an $f(m, \ell)$-approximation algorithm for \leds on bipartite graphs.
\end{lemma}
\begin{proof}
Let $(G,\ell)$, $G=(V_1\cup V_2, E)$, be the considered instance of \leds.
%Construct instance of \sciname{} with $d=2$ as follows.
For each $v \in V_1 \cup V_2$, let $\delta(v)$ be the set of edges incident to $v$. Consider the $2$-layer \iksc instance $(\mathcal{U}, k, \mathcal{S}^1, \mathcal{S}^2)$ with: $\mathcal{U} = E$, $k = \ell$, and $\mathcal{S}^i = \{\delta(v) : v \in V_i\}$ for $i=1,2$.
Note that a solution for $(G,\ell)$ translates to a solution for $(\mathcal{U}, k, \mathcal{S}^1, \mathcal{S}^2)$ and vice versa.
\qed
\end{proof}

\begin{lemma}\label{lem:dse-fli}
If there exists an $f(n, k)$-approximation algorithm for \ikmfl on two layers, then there exists an $f(m, \ell)$-approximation algorithm for \leds on bipartite graphs.
\end{lemma}

\begin{proof}
Let $(G,\ell)$, $G=(V_1\cup V_2, E)$, be the considered instance of \leds.
For each $v \in V_1 \cup V_2$, let $\delta(v)$ be the set of edges incident to $v$.
Consider the $2$-layer \ikmfl instance defined as follows. Let $\fliuniv = E$, $k=\ell$
and $\flilayer{i} = V_i$. We define all opening costs to be equal to $1$ and
all connection costs $\fliconnectl{i}(e, v)$ to be equal to $0$ if $v$ is an endpoint
of $e$, or $\infty$ otherwise. As each client (edge) $e$ is connected by a finite distance
to only one facility in each layer, costs $\fliconnectl{i}$ are metric.
Note that a solution for $(G,\ell)$ translates to a solution for $(\fliuniv, k, \flilayer{1}, \flilayer{2})$ and vice versa.
\qed
\end{proof}

\begin{lemma}\label{lem:dse-mst}
If there exists an $f(n, k)$-approximation algorithm for \ikmst on two layers, then there exists an $f(n+m+1,\ell)$-approximation algorithm for \leds\ on bipartite graphs.
\end{lemma}
\begin{proof}
Let $(G,\ell)$, $G=(V_1\cup V_2, E)$, be a bipartite instance of \leds. We show how to construct the first layer of the corresponding \ikmst\ instance: the construction of the second layer is symmetric.

Consider the following auxiliary weighted graph $G'=(V',E')$. The nodeset $V'$ is given by $\{r\}\cup V_1\cup V_2 \cup E$, where $r$ is a newly created root node. 
Moreover, $E'=E'_a\cup E'_b\cup E'_c$, where: $E'_a=\{\{r,v\} : v\in V_2\}$, $E'_b=\{\{r,v\} : v\in V_1\}$, and $E'_c=\{\{v,\{v,u\}\}: v\in V_1, \{v,u\}\in E\}$. We set the weight of edges in $E'_a$, $E'_b$, and $E'_c$ to $\infty$, $1$ and $0$, respectively. Intuitively, we want all nodes of $G'$ that correspond to $V_2$ to be too expensive to be used in this layer. Eventually, we consider the metric closure of $G'$ (set $E'$ induces a tree). We set the target $k := \ell$.

Because of the $\infty$ edges, no node $v \in V_1 \cup V_2$ will belong to the intersection. Thus the only nodes in the intersection will be the nodes corresponding to edges of the original bipartite graph. Consequently a solution for \ikmst\ of cost $\alpha$ in $G'$ translates to a solution for the \leds\ instance with $\alpha$ vertices and vice versa.
\qed
\end{proof}
Pipelining each of the Lemmas~\ref{lem:dse-sci},~\ref{lem:dse-fli} and \ref{lem:dse-mst} with Theorem~\ref{thm:dense-v2e} we prove Theorems~\ref{thm:dse-sci} and \ref{thm:dse-mst}.

\subsection{\ikufl}\label{apx:fli}

\newcommand{\flicostsum}{\overline{\fliconnect}}

%\rem{F: often d rather than h. Check!}
%\rem{F: in union the notation of opening and connection costs is different}

In this section we give a $(4k^{1-1/h} log^{1/h}(k))$-approximation algorithm for \ikufl.
The algorithm works in $N^{O(h)}$ time, i.e., a polynomial time for any fixed $h$.

Let us define the problem in a way which is convenient for our purposes. For each layer $i$, we are given a collection $\flilayer{i}$ of subsets of the set $\clients$ of clients, with one subset $X$ for each facility $f(X)$. Intuitively, $X$ is the set of clients that facility $f(X)$ is allowed to serve on layer $i$. We use $\open^i(X)$ and $\connect^i(x,X)$ as shortcuts for 
$\open^i(f(X))$ and $\connect^i(x,f(X))$, respectively.
A solution consists of a subset $\fliK \subseteq \clients$ of $k$ clients, and a subset $\flisol{i}\subseteq \flilayer{i}$. Intuitively, $\fliK$ is the subset of clients that we decide to connect, and $\cup_{X\in \flisol{i}}f(X)$ is the set of facilities that we open on layer $i$ to connect $K$ on that layer. We denote as $\flichoice{i}{x}$ the facility serving client $x$ on layer $i$.
We slightly generalize \ikufl by adding for each $x \in \fliuniv$
a cost $\fliuse(x)$ of using $x$ in the solution.
In total, the cost of a solution is:
$$\sum_{x \in \fliK} \fliuse(x) + \sum_{i=1}^h \sum_{X \in \flisol{i}} \fliopenl{i}(X) + \sum_{x \in \fliK} \sum_{i=1}^h \fliconnectl{i}\left(x, \flichoice{i}{x}\right).$$
As we shall see, this generalization does not make the problem much harder, but allows us to describe our algorithm more neatly.

%We are given $h$ collections
%$\flilayer{1}, \flilayer{2}, \ldots, \flilayer{h}$ of subsets of a given universe $\fliuniv$.
%Each $\flilayer{i}$ corresponds to a set of facilities available in the $i$-th layer.
%Each facility $X \in \flilayer{i}$ is represented as a subset of clients ($X \subseteq \fliuniv$) it can serve
%in the $i$-th layer.
%For each collection $\flilayer{i}$, for each set $X \in \flilayer{i}$
%there is a cost $\fliopenl{i}(X)$ of using $X$ (openning facility $X$)
%in the solution. Moreover, for each
%element $x \in X$ there is a cost $\fliconnectl{i}(x,X)$ of {\em{attaching}} the element $x$ to
%the set $X$. The solution constists of: a subset $\fliK \subset \fliuniv$ of size $k$,
%families $\flisol{i} \subseteq \flilayer{i}$ and an assignment (called
%a {\em{choice function}}) that, for each layer $i$ and each element $x \in \fliK$,
%chooses a set $\flichoice{i}{x} \in \flisol{i}$ that includes $x$.
%We slightly generalize \ikufl{} by adding for each $x \in \fliuniv$
%a cost $\fliuse(x)$ of using $x$ in the solution.
%In total, the cost of a solution is:
%$$\sum_{x \in \fliK} \fliuse(x) + \sum_{i=1}^h \sum_{X \in \flisol{i}} \fliopenl{i}(X) + \sum_{x \in \fliK} \sum_{i=1}^h \fliconnectl{i}\left(x, \flichoice{i}{x}\right).$$
%As we shall see, this generalization does not make the problem much harder, but allows us to describe our algorithm more neatly.

%\subsection{Case $h=1$}

A $O(\log k)$-approximation for our generalized \ikufl problem for $h=1$ follows easily from \cite{slavik}.

%First, let us solve the problem for $h=1$. This is a simple modification of the well-known $O(\log k)$ approximation algorithm for \ksc \cite{Vazirani03book}.
%\rem{F: reference missing here. For k-set cover I put Vazirani's book, but also there a precise reference is needed}

\begin{lemma}\label{approx-d-one}
There exists a $(1 + \ln k)$-approximation algorithm for the generalised \kufl problem, where connecting client $x\in \fliuniv$ has extra cost $\fliuse(x)$.
\end{lemma}

\begin{proof}
We can reduce tha generalised version of the \kufl problem to the classical one by simply increasing all distances between client $x$ and any facility $f$ by $\fliuse(x)$. The claim follows from \cite{slavik}.  
%   \in X\fliuniv$
%  and facility $f(X) \in \flilayer{}$ by the cost of opening the client $x$, formally $\fliconnect(x, X) := \fliconnect(x, X) + \fliuse(x)$.
%   To solve the classical version we use the algorithm by Slavik~\cite{slavik}.
\qed
\end{proof}

%\subsection{Arbitrary $h$}

\bigskip
\newcommand{\fliproc}{\ensuremath{\mathtt{FLI}}}
Now we solve the problem for arbitrary $h$.
Let us give some notation. 
%The threshold $\scthreshold$ is set to $\scthreshold := 4^{1-1/d} (k-|K|)^{1-1/d} \log^{1/d}(k)$.
We are going to develop a procedure $\fliproc(k, \fliuniv, h, (\flilayer{i})_{i=1}^h, (\fliopen, \fliconnect, \fliuse))$
that returns the $(4k^{1-1/h} \log^{1/h}(k))$-approximation for \ikufl.
We abandon the requirement that each $X \in \bigcup_{i=1}^h \flilayer{i}$ is contained in $\fliuniv$,
but the solution set $\fliK$ needs to be a subset of $\fliuniv$.

Note that Lemma \ref{approx-d-one} provides an algorithm for $h=1$.
For arbitrary $h$, we use the procedure described in Figure \ref{alg:approx}.
The $\fliproc$ procedure constructs set $\fliK$ with facilities $\flisol{i}$ and a choice function
greedily. At one step, we iterate over all layers $r$ and sets $X \in \flilayer{r}$
and all possible cardinalities $j$ of elements that set $X$ can cover
(i.e., $1 \leq j \leq \min(k-|\fliK|, |X \setminus \fliK|)$) and try to:
choose $X$ and solve problem without sets $\flilayer{r}$ (where $X \in \flilayer{r}$)
for universe restricted to $X \setminus \fliK$ and the goal size $j$.
We solve the subproblem using $\fliproc$ procedure, but for $h$ decreased by one.
We hide the costs of attaching elements $x \in X \setminus \fliK$ to $X$ in costs per element,
i.e., in $\fliuse(x)$.
Finally, we choose set $X$ and cardinality $j$ with the smallest cost per element covered.
This cost is kept in variable $c_{j,X}$.

%What remains to prove is the following lemma.

\begin{Figure}
\begin{minipage}{\textwidth}
\small
\begin{algorithmic}[1]
  \Procedure{$\fliproc$}{$k, \fliuniv, h, (\flilayer{i})_{i=1}^h, (\fliopen, \fliconnect, \fliuse)$}
  \If{$h=1$}
    \State{\Return solution found by Lemma \ref{approx-d-one}}
  \EndIf
  \If{$k=1$}
    \State{\Return optimal solution by brute-force}
  \EndIf
  \State{$\fliK \gets \emptyset$, $\flisol{i} \gets \emptyset$ for $i=1,2,\ldots,h$.}
  \Repeat
  \For{all $1 \leq r \leq h$ and all $X \in \flilayer{r}$}
    \For{$j := 1$ to $\min(k-|\fliK|, |X \setminus \fliK|)$}
      \State $\fliuse' \gets \fliuse$, except for elements $x \in X \setminus \fliK$, where
        we put $\fliuse'(x) = \fliuse(x) + \fliconnectl{i}(x, X)$.
      \State $(\fliK_{j,X}, (\flisol{a}_{j,X})_{a=1}^h) \gets
         \fliproc(j, X \setminus \fliK, h-1, (\flilayer{i})_{1 \leq i \leq h, i\neq r}, (\fliopen, \fliconnect, \fliuse'))$.
       \State $C_{j,X} \gets$ cost of $(\fliK_{j,X}, (\flisol{i}_{j,X})_{i=1}^h)$ w.r.t. costs $(\fliopen, \fliconnect, \fliuse')$, plus $\fliopenl{r}(X)$
      \State $c_{j,X} \gets C_{j,X}/j$.
    \EndFor
  \EndFor
  \State{$r_0, X_0, j_0 \gets$ values of the loops' iterators for which the cheapest solution was found
    according to weights $c_{j,X}$.}
  \State{$\fliK \gets \fliK \cup \fliK_{j_0,X_0}$, $\flisol{r_0} \gets \flisol{r_0} \cup \{X_0\}$,
    $\flisol{i} \gets \flisol{i} \cup \flisol{i}_{j_0, X_0}$ for $i \neq r_0$}
  \State{For each $x \in \fliK_{j_0,X_0}$ assign $\flichoice{r_0}{x} = X_0$.}
  \Until{$|\fliK| = k$}
  \State{\Return $(\fliK, (\flisol{i})_{i=1}^h)$}
  \EndProcedure
\end{algorithmic}
\end{minipage}
\caption{Approximation algorithm for \ikufl.}\label{alg:approx}
\end{Figure}

%\begin{lemma}\label{approx-thm}
%  Algorithm described in Pseudocode \ref{alg:approx}
%  returns a solution within a multiplicative factor $(4k^{1-1/d}\log^{1/d}(k))$
%  of the optimal solution.
%\end{lemma}

\begin{proof} {\em (Theorem \ref{thr:intersection:sc2})}
We consider the algorithm in Figure \ref{alg:approx}. It is clear that $\fliproc$ procedure works in $N^{O(h)}$ time, as it calls
recursively itself $O(Nk^2)$ times with $h$ decreased by one.

We next prove the claim on the approximation by induction on $h$. For $h=1$ the thesis is implied by Lemma \ref{approx-d-one}. Assume now that all recursive calls in the $\fliproc$ procedure return solutions with $(4k^{1-1/(h-1)}\log^{1/(h-1)}(k))$ approximation ratio.

  Let $\opt$ be the cost of the optimal solution for the given instance.
  Pick any optimal solution with cost $\opt$ and let $\flioptK$ be the set of covered elements by it,
  and $\flioptsol{i}$ be the chosen subset of $\flilayer{i}$ for $i=1,2,\ldots,h$.
  For each layer $i$ and each element $x \in \flioptK$ we fix the set $\flioptchoice{i}{x} \in \flioptsol{i}$ that covers $x$ in the optimal solution.

  We prove that at one step of the algorithm, the weight $c_{j_0, X_0}$ satisfies:
  $$c_{j_0,X_0} \leq \opt \cdot 4^{1-1/h}(k-|\fliK|)^{-1/h}\log^{1/h}(k).$$
  Recall that $c_{j_0,X_0}$ is the average cost
  paid newly covered elements.
  This bound is sufficient, since the total cost of the constructed solution is bounded by:
  \begin{align*}
    \textrm{total cost} &\leq \opt \cdot \sum_{i=0}^{k-1} 4^{1-1/h}(k-i)^{-1/h}\log^{1/h}(k) 
      \leq \opt \cdot 4^{1-1/h}\log^{1/h}(k) \int_0^k x^{-1/h} dx \\
      &= \opt \cdot 4^{1-1/h}\log^{1/h}(k) \frac{1}{1-1/h} k^{1-1/h} 
      \leq \opt \cdot 4\log^{1/h}(k) k^{1-1/h}.
  \end{align*}
  The last inequality follows from the fact that $4^{-\varepsilon} \leq 1-\varepsilon$
  for $0 \leq \varepsilon \leq \frac{1}{2}$.

  Let $\scthreshold := 4^{1-1/h} (k-|\fliK|)^{1-1/h} \log^{1/h}(k)$.
  We consider two cases, depending on whether there exists a layer $r$ and
  a set $X \in \flioptsol{r}$
  that covers at least $\scthreshold$ elements of $\flioptK \setminus \fliK$, i.e,
  $$|\{x \in \flioptK \setminus \fliK: \flioptchoice{r}{x} = X\}| \geq \scthreshold.$$

  {\bf Case 1.} Assume there exists a layer $r$ and a set $X \in \flioptsol{r}$
  such that for at least $\scthreshold$ elements $x$ of
  $\flioptK \setminus \fliK$ we have $\flioptchoice{r}{x} = X$.
  Let us focus on the moment when our algorithm considers taking set $X$.
  We may assume $\scthreshold \leq k-|\fliK|$, as otherwise $k-|\fliK|$ is bounded by constant
  and we may instead use brute force to finish the greedy construction
  optimally.
  Therefore our algorithm considers covering $\scthreshold$ elements of $X$.
  As the optimal solution does it, it may be done with cost $\opt$, so the
  recursive call returns the solution with cost at most
  $\opt \cdot 4\scthreshold^{1-1/(h-1)} \log^{1/(h-1)}(\scthreshold)$. We cover $\scthreshold$ elements, so
  \begin{align*}
    c_{\scthreshold,X} &\leq \opt \cdot 4\scthreshold^{-1/(h-1)} \log^{1/(h-1)}(\scthreshold) \\
    &\leq \opt \cdot 4\left( 4^\frac{h-1}{h} (k-|\fliK|)^\frac{h-1}{h} \log^{1/h}(k) \right)^\frac{-1}{h-1} \log^{1/(h-1)}(k) \\
    &= \opt \cdot 4^{1-1/h}(k-|\fliK|)^{-1/h} \log^{1/h}(k).
  \end{align*}

  {\bf Case 2.} Every $X \in \fami'_{\opt}$ covers at most $\scthreshold$ elements of $\flioptK \setminus \fliK$.
  For each $x \in \flioptK \setminus \fliK$ denote
  $$\flicostsum(x) = \fliuse(x) + \sum_{r=1}^h \fliconnectl{r}(x, \flioptchoice{r}{x}) + 
    \fliopenl{r}(\flioptchoice{r}{x}),$$
  i.e., the total cost of choosing $x$, attaching it to set $\flioptchoice{r}{x}$ and
  choosing set $\flioptchoice{r}{x}$. By the assumption in this case, we have
  \begin{align*}
    \sum_{x \in \flioptK \setminus \fliK} \flicostsum(x) &=
  \sum_{x \in \flioptK \setminus \fliK} \fliuse(x) + \sum_{r=1}^h \fliconnectl{r}(x, \flioptchoice{r}{x}) \\
    & \qquad+ \sum_{r=1}^h \sum_{X \in \flioptsol{r}} \fliopenl{r}(X) \cdot |\{x \in \flioptK \setminus \fliK: \flioptchoice{r}{x} = X\}| \leq \scthreshold \cdot \opt.
  \end{align*}
  Thus there exists $x_0 \in \flioptK \setminus \fliK$
  such that $\flicostsum(x_0) \leq \scthreshold \cdot \opt / |\flioptK \setminus \fliK|$.
  Let $X_0 = \flioptchoice{h}{x_0}$.
  Note that, since our algorithm uses brute force for $k=1$,
  the recursive call with find optimal solution for $j := 1$ and $X := X_0$ and thus
  $c_{1,X_0} \leq \flicostsum(x_0)$. As $|\flioptK \setminus \fliK| \geq k-|\fliK|$ we have:
  \begin{align*}
    c_{1,X_0} &\leq \scthreshold \cdot \opt / |\flioptK \setminus \fliK| \\
    & \leq \opt \cdot 4^{1-1/h} (k-|\fliK|)^{1-1/h} \log^{1/h}(k) (k-|\fliK|)^{-1} \\
    & = \opt \cdot 4^{1-1/h} (k-|\fliK|)^{-1/h} \log^{1/h}(k).
  \end{align*}
\qed
\end{proof}

\section{Union Problems}\label{apx:union}

In this section we give the omitted details concerning union covering problems.

\subsection{A Greedy Approach}
\label{apx:union:greedy}

We next describe a simple greedy algorithm which provides a logarithmic approximation for several union partial covering problems. Consider a partial covering problem where $\requests$ is the set of requests, $\items^i$ is the set of items on layer $i$, with costs $\weight^i:\items^i\to \R_{\geq 0}$, and $k$ is the target. We require that the covering problem satisfies a natural \emph{composition} property, namely two solutions satisfying $k'$ and $k''$ distinct requests, can be merged (without increasing the total cost) to obtain a solution satisfying $k'+k''$ requests. (Merging might involve some polynomial-time operations). The algorithm works as follows:
\begin{enumerate}\itemsep0pt
\item[(1)] For all layers $i$, for all $k^i:=1,\ldots,k$, solve the single-layer problem induces by the triple $(\requests,\items^i,k^i)$ with a $\rho$-approximation algorithm.
\item[(2)] Among all the solutions computed, take the one $\approximate$, obtained for some triple $(\requests,\items^i,k^i)$, which minimizes the ratio $\weight^i(\approximate)/k^i$. 
\item[(3)] Merge $\approximate$ with the solution under construction. Remove from $\requests$ the requests satisfied by $\approximate$, and decrease $k$ by $k^i$. 
\item[(4)] If $k>0$, go to Step (1). Otherwise return the current solution. 
\end{enumerate}
\begin{theorem}\label{thr:greedy}
The algorithm above computes a $O(\rho\, \log k)$-approximation for the partial covering problem considered in polynomial time.
\end{theorem}
\begin{proof}
The claim on the running time is trivial. The algorithm computes a feasible solution, due to the composition property. Consider now the approximation ratio. Let $\approximate_1,\ldots,\approximate_q$ be the sequence of approximate solutions computed, $w_j$ be the cost of $\approximate_j$ and $k_j$ the number of requests that it satisfies on layer $i_j$. Observe that at the beginning of iteration $j$, the current number of requests is $k-\sum_{a<j}k_{a}$, and the cost of  the optimal solution with respect to that number of requests is no more than $\opt$. By an averaging argument, at each iteration $j$ we have $\weight^{j_i}(\approximate_j)/k_j\leq \opt/(k-\sum_{a<j}k_j)$. We can conclude that the cost of the solution computed is at most
$$
\rho\,\opt \left( \frac{k_1}{k}+\frac{k_2}{k-k_1}+\ldots+\frac{k_q}{k-\sum_{a<q}k_a} \right)\leq \rho\,\opt\cdot \ln k.
$$
%$$
%\rho\opt\left\( \frac{k_1}{k}+\frac{k_2}{k-k_1}+\ldots+\frac{k_q}{k-\sum_{a<q}k_a} \right\) \leq \rho\opt\cdot \ln k.
%$$
\qed
\end{proof}
\begin{corollary}
There are $O(\log k)$-approximation algorithms for \ukmst and  \ukmfl.
\end{corollary}
\begin{proof}
Observe that removing requests transforms the original \kmst problem in each layer into a \kst problem: for the latter problem there is a $4$-approximation algorithm \cite{Garg05stoc}.
Note also that all the partial solutions in layer $i$ contain the root $r^i$: hence the merging step is trivial. The claim for \ukmst follows.

%Due to triangle inequality, it is easy to turn the mentioned $3$-approximation for \kmst into a $6$-approximation for (rooted) \ktsp. Triangle inequality also implies the same approximation ratio for the generalization of the problem where only some terminal nodes contribute to reach the target (simply discard non-terminal nodes): this way we can remove satisfied requests. Eventually, given a set of cycles containing the root, we can obtain a simple cycle spanning the same set of nodes via shortcutting: this gives the needed merging step. The claim follows for \uktsp. 

For \kmfl, there is a $2$-approximation algorithm in \cite{JMMSV03}. In this case removing a request simply means removing one client, and the merging step is trivial. This proves the claim for \ukmfl. 
\qed
\end{proof}

\subsection{\ukmfl}
\label{apx:union:extensions}

In this section we present an LP-based $O(h)$-approximation algorithm for \ukmfl. As we will see, the basic idea is an for \ukmst.

Recall that in \ukmfl\  we are given a graph $G=(V,E)$, a set $\clients\subseteq V$ of clients, a set $\facilities\subseteq V$ of facilities, one integer $k$ (target), a set of opening cost functions $\open^i:\facilities\to \R_{\geq 0}$, and a set of edge-weight functions $\connect^i:E\to \R_{\geq 0}$, with $i=1,\ldots,h$. The distance between nodes $u$ and $v$ w.r.t. $\connect^i$ is denoted as $\connect^i(u,v)$. A feasible solution is given by a pair $(\tilde{\clients}^i,\tilde{\facilities}^i)$ for each layer $i$, $\tilde{\clients}^i\subseteq \clients$ and $\tilde{\facilities}^i\subseteq \facilities$, such that $|\cup_i \tilde{\clients}^i|\geq k$.
The goal is minimizing the cost $\sum_{i=1,\ldots,h}(\sum_{f\in \tilde{\facilities}^i}\open^i(f)+\sum_{c\in \tilde{\clients}^i}\connect^i(c,\tilde{\facilities}^i))$.
Here $w^i(c,\facilities')$ denotes the minimum distance on layer $i$ between client $c\in \clients$ and facility $f\in \facilities'\subseteq \facilities$.

Also in this case we consider a natural LP relaxation $LP_{kMFL}(\clients,\facilities,\open,\connect,k)$ for the single-layer version of the problem:
\begin{align*}
\min\quad & \textstyle{\sum_{f\in \facilities}o(f)y_f+\sum_{(c,f)\in \clients\times \facilities}\connect(c,f)\,x_{c,f}} &  &\\
 s.t.\quad  & \textstyle{x_{c,f}\leq y_f}, & \textstyle{\forall (c,f)\in \clients\times \facilities;} & \\
        & \textstyle{\sum_{f\in \facilities}x_{c,f} \geq z_c,} & \textstyle{\forall c\in \clients;} & \\
        & \textstyle{\sum_{c\in \clients}z_{c} \geq k;} &  & \\
        & \textstyle{x_{c,f},y_f\geq 0,1\geq z_c\geq 0,} & \textstyle{\forall c\in \clients,\forall f\in \facilities.} &
\end{align*}
Variable $y_f$ indicates whether facility $f$ is opened, and variable $x_{c,f}$ whether client $c$ is connected to facility $f$. Variable $z_c$ indicates whether client $c$ is connected to  some facility. We need the following result.
\begin{lemma}\label{lem:charikar} \cite{CKMN01soda}
Let $(\clients,\facilities,\open,\connect,k)$ be an instance of \kmfl, $o_{max}:=\max_{f\in \facilities}\open(f)$, and $\opt'$ be the optimal solution to $LP_{kMFL}(\clients,\facilities,\open,\connect,k)$. There is a polynomial time algorithm {\tt apx-kmfl} which computes a solution to the instance of cost at most $3\opt'+2o_{max}$.
%There is a polynomial time algorithm {\tt apx-kmfl} which computes a solution to the $k$-mFL problem of cost at most $3opt_{LP}+2f_{max}$, where $opt_{LP}$ is the optimal solution to $LP_{kmFL}$ and $f_{max}$ is the maximum cost of a facility.
\end{lemma}
The algorithm that we use is analogous to the one for the $k$-MST case. In a preliminary phase we guess the largest cost $\open^*$ of a facility in the optimum solution, and remove all facilities of larger cost. Let $\facilities^i$ be the remaining set of facilities on layer $i$. We then compute the optimal solution $\OPT^*=(x^i,y^i,z^i,z)_{i}$, of cost $\opt^*$, to the following relaxation $LP_{ukMFL}$ for the problem:
\begin{align*}
\min\quad & \textstyle{\sum_{i=1,\ldots,h}(\sum_{f\in \facilities^i}\open^i(f)y^i_f+\sum_{(c,f)\in \clients\times \facilities^i}\connect^i(c,f)\,x^i_{c,f})} &  &\\
 s.t.\quad  & \textstyle{x^i_{c,f}\leq y^i_f,} & \textstyle{\forall i\in \{1,\ldots,h\},\forall (c,f)\in \clients\times \facilities^i;} & \\
        & \textstyle{\sum_{f\in \facilities^i}x^i_{c,f} \geq z^i_c,} & \textstyle{\forall i\in \{1,\ldots,h\},\forall c\in \clients;} & \\
        & \textstyle{\sum_{i=1,\ldots,h }z^i_{c} \geq z_c,} & \textstyle{\forall c\in \clients;} & \\
        & \textstyle{\sum_{c\in \clients}z_{c} \geq k;} &  & \\
        & \textstyle{x^i_{c,f},y^i_f,z^i_c\geq 0,1\geq z_c\geq 0,} & \textstyle{\forall i\in \{1,\ldots,h\},\forall c\in \clients,\forall f\in \facilities.} &
\end{align*}
%The algorithm that we use is analogous to the one for the $k$-MST case. In a preliminary phase we guess the most expensive facility in the optimum solution, and remove all facilities of larger cost. Then compute the optimal solution $OPT^*=\{x^i_{c,f},y^i_f,z^i_c\}_{c,f,i}$, of cost $opt^*$, to $LP_{ukmFL}$. 
Then we identify for each layer $i$ the subset of clients $\clients^i:=\{c\in \clients: z^i_c=\max_{j}\{z^j_c\}\}$. We run the algorithm {\tt apx-kmfl} from Lemma \ref{lem:charikar} on each layer, with clients $\clients^i$, facilities $\facilities^i$, and target $k^i:=\lfloor\sum_{c\in \clients^i}z_c\rfloor$: we open facilities and connect clients accordingly. Let $k'$ be the number of connected clients. If $k'<k$, we connect extra clients in a greedy fashion, possibly opening new facilities: in particular, we consider the pairs $(c,f)\in \clients^i\times \facilities^i$, with $c$ not connected, which minimize $o^i(f)+w^i(c,f)$, and we connect the corresponding clients.
%we consider the unconnected clients which are closer to open facilities, as well as pairs $(c,f)$, where $c$ is an unconnected client and $f$ is a facility not opened.
\begin{theorem}\label{thr:constantkSteiner}
There is a $O(h)$-approximation algorithm for \ukmfl. The running time of the algorithm is $O((nh)^{O(1)})$.
\end{theorem}
\begin{proof}
Consider the algorithm above. The claim on the running time is trivial. As in the \kmst case, consider the feasible fractional solution $(\tilde{x}^i,\tilde{y}^i,\tilde{z}^i,z)_{i}$ obtained from $\OPT^*$ by setting $\tilde{z}^i_c=z_c$ if $c\in \clients^i$, $\tilde{z}^i_c=0$ otherwise, and raising the variables $x$ and $y$ by a factor $h$. This new solution costs at most $h\cdot \opt^*$. Furthermore, $(\tilde{x}^i,\tilde{y}^i,\tilde{z}^i)$ is a feasible fractional solution to $LP_{kMFL}(\clients^i,\facilities^i,\open^i,\connect^i,k^i)$. Let $\tilde{apx}^i$ be its cost. Lemma \ref{lem:charikar} guarantees that the cost of the integral solution on layer $i$ is at most $3\tilde{apx}^i+2\open^*$. Since $\sum_i k^i\geq k-h$, the final step costs at most $h\cdot \opt$. Altogether the cost of the solution computed is at most
$$
\sum_i (3\tilde{apx}^i+2o^*)+h\cdot \opt\leq 3h\cdot \opt^*+2h\cdot o^*+h\cdot \opt\leq 6h\cdot \opt.
$$
\qed
\end{proof}
Also in this case a more technical analysis allows one to reduce the constant in front of $h$ in the approximation factor, at the cost of a larger running time.
%\begin{corollary}
%There are a $O(h)$-approximation algorithm for unrooted \ukst, \uktsp, and \ukmfl. The running time of these algorithms is $O((nh)^{O(1)}n^h)$.
%\end{corollary}

\subsection{Integrality Gap}
\label{apx:union:gap}

%We next show that the $O(h)$-approximations obtained for \ukmst and \ukmfl are tight. In order to do that, it is sufficient to bound the integrality gap of $LP_{ukMST}$ and $LP_{ukMFL}$.\rem{F: is this really sufficient?}
\begin{lemma}
The integrality gap of $LP_{ukMST}$ and $LP_{ukMFL}$ is $\Omega(h)$.
\end{lemma}
\begin{proof}
We consider the following unweighted \setc instance given in \cite{Vazirani03book}. Let $G'$ be an hypergraph on $m'$ nodes, which has one hyperedge for any subset of $m'/2$ nodes. We construct a set cover instance with $m'$ sets given by nodes, $\binom{m'}{m'/2}$ elements given by hyperdeges, and inclusion given by incidence. Taking a fraction $2/m'$ of each set gives a feasible fractional solution of cost $2$ to the natural set cover LP. On the other hand, the  optimal integral solution uses $m'/2+1$ sets. Hence the integrality gap in this case is $\Omega(m')$.

The same reductions as in Theorem \ref{thr:hardnessRooted} imply a $\Omega(h)$ lower bound on the integrality gap of $LP_{ukMST}$ and
$LP_{ukMFL}$ for the case $k=n$.
%In particular, assigning value $2/m'=2/h$ to the variables $x^i_e$ and $z^i_e$ associated to sets provides a feasible fractional solution to $LP_{ukMST}$ of cost $2$, while the optimal integral solution costs $m'/2+1=h'/2+1$.
\qed
\end{proof}

\end{document}